\documentclass[12pt]{iopart}

\newcommand{\R}{\mathbb{R}}

\usepackage{amssymb}
\usepackage{hyperref}
\usepackage{graphicx}
\usepackage{iopams,amsthm}
\usepackage{color}
\usepackage{cancel}
\usepackage{url}
\newtheorem{thm}{Theorem}[section]

\newtheorem{lem}{Lemma}[section]

\newtheorem{definition}{Definition}[section]

\begin{document}
\title[Inverse diffraction for {\em{SDO/AIA}}]
{Inverse diffraction for the {\em{Atmospheric Imaging Assembly}} in the {\em{Solar Dynamics Observatory}} }
\author{G Torre\dag, R A Schwartz\ddag, F Benvenuto\S, A M Massone\$ and M Piana\dag \$}
\address{\dag Dipartimento di Matematica, Universit\`a di Genova, via Dodecaneso 35, 16146 Genova, Italy \\
\ddag Catholic University of America and NASA Goddard Space Flight Center, Greenbelt, MD 20771, USA \\ 
\S Ecole Polytechnique, CMAP, Route de Saclay, 91128 Palaiseau Cedex, France \\
\$ CNR - SPIN Genova, via Dodecaneso 33, 16146 Genova, Italy}
\ead{piana@dima.unige.it}

\begin{abstract}
The {\em{Atmospheric Imaging Assembly}} in the {\em{Solar Dynamics Observatory}} provides full Sun images every $12$ seconds in each of $7$  Extreme Ultraviolet passbands. However, for a significant amount of these images, saturation affects their most intense core, preventing scientists from a full exploitation of their physical meaning. In this paper we describe a mathematical and automatic procedure for the recovery of information in the primary saturation region based on a correlation/inversion analysis of the diffraction pattern associated to the telescope observations. Further, we suggest an interpolation-based method for determining the image background that allows the recovery of information also in the region of secondary saturation (blooming).  
\end{abstract}

\maketitle

\section{Introduction}

The {\em{Solar Dynamics Observatory (SDO)}} \cite{pethch12} is a solar satellite launched by NASA on February 11 2010. The scientific goal of this mission is a better understanding of how the solar magnetic field is generated and structured and how solar magnetic energy is stored and released into the helio- and geo-sphere, thus influencing space weather. {\em{SDO}} contains a suite of three instruments: 
\begin{itemize}
\item The {\em{Helioseismic and Magnetic Imager (SDO/HMI)}} \cite{scetal12} has been designed to study oscillations and the magnetic field at the solar photosphere. 
\item The {\em{Atmospheric Imaging Assembly (SDO/AIA)}} \cite{leetal12} is made of four telescopes, providing ten full-Sun images every twelve seconds, twenty four hours a day, seven days a week. 
\item The {\em{Extreme Ultraviolet Variability Experiment (SDO/EVE)}} \cite{woetal12} measures the solar extreme ultraviolet (EUV) irradiance with unprecedented spectral resolution, temporal cadence, accuracy, and precision.
\end{itemize}

The present paper deals with an important aspect of the image reconstruction problem for {\em{SDO/AIA}} \cite{gretal11,boetal12,poetal13}. The four telescopes of such instrument capture images of the Sun's atmosphere in ten separate wave bands, seven of which centered at EUV wavelengths. Each image is a $4096 \times 4096$ square array with pixel width in the range $0.6-1.5$ arcsec and is acquired according to a standard CCD-based imaging technique. In fact, each {\em{AIA}} telescope utilizes a $16$-megapixel CCD divided into four $2048 \times 2048$ quadrants. As typically happens in this kind of imaging, {\em{AIA}} CCDs are affected by primary saturation and blooming, which degrade both quantitatively and qualitatively the {\em{AIA}} imaging properties. {\em{Primary saturation}} \cite{makl97} refers to the condition where a set of pixel cells reaches the Full Well Capacity, i.e. these pixels store the maximum number possible of photon-induced electrons. At saturation, pixels lose their ability to accommodate additional charge, which therefore spreads into neighboring pixels, causing either erroneous measurements or second-order saturation. Such spread of charge is named {\em{blooming}} \cite{makl97} and typically shows up as a bright artifact along a privileged axis in the image. Figure \ref{fig:saturation-blooming} shows a notable example of combined saturation and blooming effects in an {\em{SDO/AIA}} image captured during the September 6, 2011 event. The recovery of information in the primary saturation region by means of an inverse diffraction procedure is the main goal of the present paper. Further, we also introduce here an interpolation approach that allows a robust estimate of the background in the diffraction region as well as reasonable estimate of the flux in the central blooming region.
%
%
\begin{figure}[pht]
\begin{center}
\begin{tabular}{cc}
\includegraphics[width=7.cm]{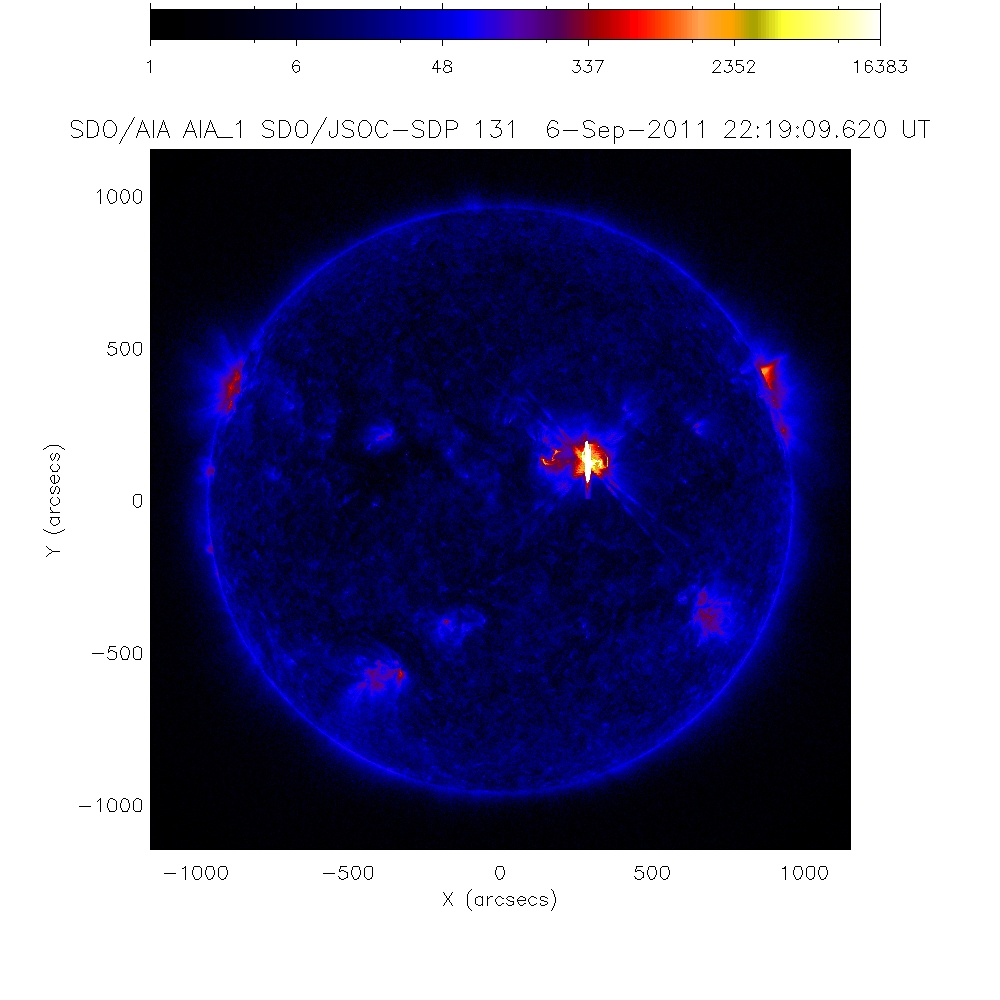} & 
\includegraphics[width=7.cm]{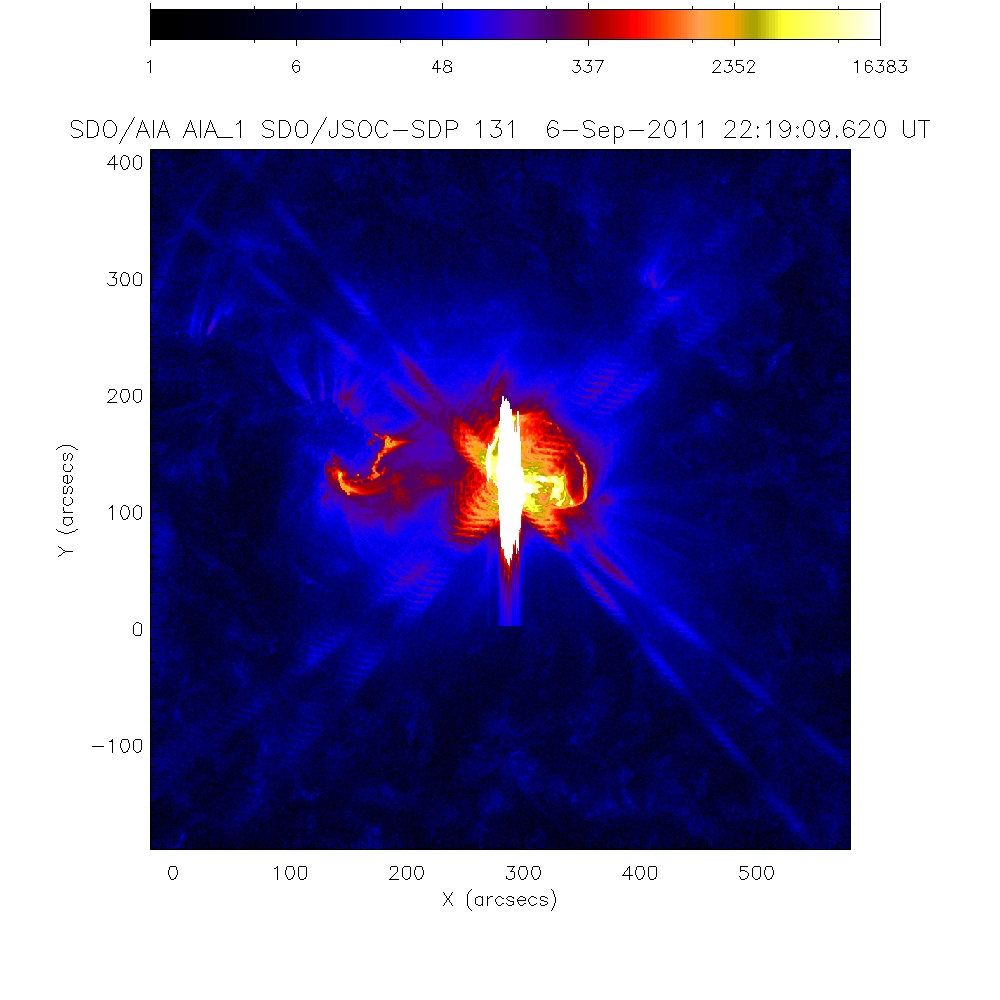} \\
\end{tabular}
\caption{{\em{SDO/AIA}} image of the September 6 2011 event, captured by means of the $131 \AA$ passband at 22:19:09 UT for $\sim 2.9$ sec exposure duration. The left panel shows the position of the explosion on the full disk of the Sun. The zoom in the right panel clearly illustrates the presence of primary saturation, blooming and diffraction fringes.}
\label{fig:saturation-blooming}
\end{center}
\end{figure}

The optical setup of each {\em{AIA}} telescope is characterized by structures with uniform wire meshes used to support the thin filters that create the EUV passbands. The interaction between the incoming EUV radiation and the grids generates a diffraction effect \cite{beba78} depending on the source intensity $I_0$. This effect can be easily observed when $I_0$ is so large that the diffracted flux is comparable to the background level.
%
When this situation occurs, it often happens that the incoming flux generates a signal that exceeds the saturation level of the {\em{AIA}} CCDs, which is $16383$ DN pixel$^{-1}$ (where DN stands for Data Number). This fact has a very interesting mathematical implication: all information on the radiation flux which is lost due to primary saturation is actually present in the diffraction pattern and therefore the signal in the primary saturation region can in principle be restored by solving the inverse diffraction problem (we point out that pixels contained in the blooming region generate diffraction effects that are negligible with respect to the ones corresponding to the primary saturation region).

The present paper addresses the de-saturation problem for {\em{SDO/AIA}} having been inspired by the heuristic and semi-automatic approaches described in \cite{rakrli11} and \cite{sctopi14}, respectively. More specifically, we describe here a fully automatic numerical method that utilizes correlation and statistical regularization to recover the image information in the primary saturation region and applies interpolation to ameliorate the effects of blooming. The first pillar of our approach is the definition of a forward model for {\em{SDO/AIA}} data formation. In fact the point spread function (PSF) of each passband can be approximated as the sum of a core PSF, which is modeled by a two-dimensional Gaussian function, and a diffraction PSF, that describes the diffraction pattern corresponding to a point source. Therefore the forward model is encoded by the linear integral operator whose integral kernel is the sum of the two PSFs. However, the difficult issue here is to define the domain of the diffraction PSF, i.e. to automatically segment the primary saturation region with respect to the blooming one. We solved this problem by combining correlation and thresholding. Once the forward operator is defined, the second pillar of our approach performs image reconstruction by means of an Expectation Maximization (EM) algorithm \cite{shva82} applied to the inverse diffraction problem. EM is a maximum likelihood technique working when the measured data are affected by Poisson noise and the solution to reconstruct is non-negative. We note that {\em{AIA}} data are only approximately Poisson, since the system only records the charge and then divides it by the average charge per photon to get the Data Number. On the other hand, the source distribution in the saturated region, is certainly positive. Further, in this paper we will take advantage of a very effective, recently introduced stopping rule for EM, which guarantees the right amount of regularization by means of a criterion with solid statistical basis \cite{bepi14}. The effectiveness of this de-saturation method is verified against synthetic data and by reconstructing the source distribution in the saturated regions of {\em{SDO/AIA}} maps of the September 6 2011 event, acquired at different time points.

The plan of the paper is as follows. Section 2 models the forward problem. Section 3 describes the image reconstruction method utilized for the solution of the inverse diffraction problem. Section 4 performs a numerical validation of the de-saturation approach in the case of synthetic data mimicking the {\em{SDO/AIA}} signal formation process. An example of how the method works in the case of real data is described in Section 5. Finally, our conclusions are offered in Section 6 while an appendix illustrates the way we estimate the background and recover information in the blooming region.

\section{The forward problem}
As shown in the Introduction, during many observations an {\em{SDO/AIA}} image presents a rather complex structure. Using a lexicographic order for the image pixels and if $I$ with size $N$ is one of such images, then in $I$ it is possible to point out five different sets of pixels:
\begin{enumerate}
\item The set of saturated pixels
\begin{equation}\label{set-saturated}
S^{\prime} = \{ i \in {\mathcal{N}},~ I_i = 16383~ {\mbox{DN pixel}}^{-1}\}~,
\end{equation}
where ${\mathcal{N}}$ is the set of natural numbers ranging from 1 to $N$.
\item The subset $S \subset S^{\prime}$ of saturated pixels which are affected by primary saturation.
\item The subset $B \subset S^{\prime}$ of saturated pixels which are affected by blooming. Of course we have that 
\begin{equation}\label{union}
S^{\prime} = S \cup B~~~~~~~S \cap B = \emptyset~.
\end{equation}
\item The set of pixels $F$ such that  $F \cap S^{\prime} = \emptyset$ and where the diffraction fringes occur.
\item The complement of $S^{\prime} \cup F$ (as we will see this set does not play any role in the de-saturation process).
\end{enumerate}

In general, the data formation process in {\em{AIA}}, which generates the image $I$, is the result of the discretization of the convolution between the telescope PSF and the incoming photon flux. In a finite-dimension setting, we can introduce the matrix
\begin{equation}\label{psf}
A = A_D + A_C~,
\end{equation}
which is the sum of the two $N \times N$ circulant matrices $A_D$, associated to the diffraction component of the {\em{AIA}} PSF, and $A_C$, associated to the diffusion component. Therefore the image $I$ is given by
\begin{equation}\label{psf-1}
I = A {\tilde{x}} = A_D {\tilde{x}} + A_C {\tilde{x}}~,
\end{equation}
where ${\tilde{x}}$ is the vector of dimension $N$ obtained by discretizing the incoming photon flux. We now define the sub-matrix $A_D^S: \R^{\#S} \rightarrow \R^{\#F}$ of $A_D$ that maps the vector $x$ of the values of the photon flux coming just from $S$ onto the vector made of the diffraction fringes; here $\#S$ and $\#F$ are the cardinality of $S$ and $F$, respectively. Since the diffraction effects are negligible for pixels outside the core $S$, we have that 
the diffraction pattern $I_F = \{ I_i ~, ~ i \in F\}$ is approximated by the matrix times vector product
\begin{equation}\label{cyclic}
I_F = A_D^S x + BG_F~,
\end{equation}
where $BG := A_C {\tilde{x}}$ is the total background and $BG_F := (A_C{\tilde{x}})_F$ is its restriction onto $F$.
This equation represents the forward model for {\em{AIA}} imaging we are interested in in this paper, and is completely defined once $S$, $F$ and $BG_F$ are explicitly estimated. We point out that this model neglects the diffraction of the background region on itself as well as the diffraction of the bloomed region. It follows that, in this context, the overall background $BG$ is the image deprived by the diffraction effects.

We first observe that $F$ is directly related to $S$. In fact, if $S$ is known, then 
\begin{equation}\label{effe}
F= \{ i \in {\mathcal{N}} \setminus S^{\prime}~,~(A_D \chi_S)_i \neq 0\}~,
\end{equation}
where $\chi_S$ is a vector with size $N$ whose components are $1$ in $S$ and zero elsewhere. Equation (\ref{effe}) points out the set of pixels outside the saturation region $S^{\prime}$, illuminated by the diffraction pattern produced by point sources located in the primary saturation region $S$.
On the other hand, if $\chi_{S^{\prime}}$ is a vector with size $N$ whose components are $1$ in $S^{\prime}$ and zero elsewhere,
\begin{equation}\label{effeprimo}
F^{\prime}= \{ i \in {\mathcal{N}} \setminus S^{\prime}~,~(A_D \chi_{S^{\prime}})_i \neq 0\},
\end{equation}
would point out the set of pixels outside $S^{\prime}$ illuminated by the diffraction pattern produced by point sources located in the overall saturation region $S^{\prime}$.
Based on the observation that the diffraction effects associated to blooming are negligible with respect to the ones associated to primary saturation, equations (\ref{effe}) and (\ref{effeprimo}) suggest that the segmentation of $S$ in $S^{\prime}$ can be obtained by a simple correlation analysis. In fact,
let us consider first the ideal case where no background affects the fringe data $I_F$; then the correlation vector is the back-projection
\begin{equation}\label{correlation-1}
C = (A_{D}^{S^{\prime}})^T I_{F^{\prime}},
\end{equation}
where 
$I_{F^{\prime}} = \{ I_i ~, ~ i \in F^{\prime}\}$ and $ (A_D^{S^{\prime}})^T$ is the transpose of the matrix $A_{D}^{S^{\prime}}$, i.e. the sub-matrix of $A_D$ that maps the vector of the photon flux coming from  $S^{\prime}$ on the image data in $F^{\prime}$ (the size of $A_{D}^{S^{\prime}}$ is $\#F^{\prime} \times \#S^{\prime}$). Given $C$ in (\ref{correlation-1}), $S$ would be identified by the positions associated to the components of $C$ larger than $16383$ DN pixel$^{-1}$. We now observe that, if $A_{D}^{S^{\prime}}$ is normalized as
\begin{equation}\label{normalization}
\sum_{i=1}^{\#F^{\prime}} (A_{D}^{S^{\prime}})_{ij} = 1~~~~ j=1,\ldots,\#S^{\prime},
\end{equation}
then equation (\ref{correlation-1}) can be interpreted as the first iteration of Expectation-Maximization (EM) in absence of background and with initialization given by a unit vector. Therefore, in order to account for the presence of background, we generalize the computation of the correlation by using 
\begin{equation}\label{correlation-2}
C^{(1)} = C^{(0)} \cdot (A_{D}^{S^{\prime}})^T \frac{I_{F^{\prime}}}{A_{D}^{S^{\prime}}C^{(0)} + BG_{F^{\prime}}} ~~,
\end{equation}
which is the first EM iteration when $C^{(0)}$ is a generic initialization vector, $A_{D}^{S^{\prime}}$ is not normalized as in (\ref{normalization}) and the vector $BG_{F^{\prime}}$ containing the background 
values in $F^{\prime}$  is included (we notice that in (\ref{correlation-2}) and from now on, the symbol $\cdot$ and the fraction should be intended as element-wise). It follows that, in order to explicitly compute $C^{(1)}$ in (\ref{correlation-2}) we need to estimate the background vector $BG_{F^{\prime}}$ and to select $C^{(0)}$. In the 
Appendix we show a simple way to estimate the background on the whole image. From it a reasonable choice for the initialization is $C^{(0)} = BG_{S^{\prime}}$, i.e. the background values in $S^{\prime}$. Once computed $C^{(1)}$ via (\ref{correlation-2}), the primarily saturated region $S$ is given by the positions of $C^{(1)}$ whose corresponding components are more intense than $16383$, the corresponding $F$ is given by equation (\ref{effe}) and the forward model (\ref{cyclic}) is completely defined.

\section{Inversion method for de-saturation}
The identification of the saturation region by means of correlation allows the definition of the {\em{SDO/AIA}} de-saturation problem as the linear inverse problem of determining $x$ from the measured data $I_F$ and an estimate $BG_F$ of the background in $F$, when $x$, $I_F$ and $BG_F$ are related by (\ref{cyclic}).
Since the noise affecting the measured data $I_F$ has an approximate Poisson nature, the likelihood function, i.e. the probability of obtaining the realization $I_F$ of the data random vector given the realization $x$ of the solution random vector can be written as
\begin{equation}\label{poisson}
p(I_F|x) = \prod_{ i=1}^{\#F} \frac{e^{-(A_D^Sx + BG_F)_i}}{(I_F)_i !} (A_D^Sx + BG_F)_i^{(I_F)_i}.
\end{equation}
A classical statistical approach to the solution of (\ref{cyclic}) considers the constrained Maximum Likelihood problem
\begin{equation}\label{ML-constrained}
\max_x p(I_F|x)~~~~~|~~~~~x \geq 0.
\end{equation}
Expectation Maximization (EM) solves this problem iteratively by means of \cite{shva82}
\begin{equation}\label{EM}
x^{(k+1)} = x^{(k)} \cdot (A_D^S)^T \left(\frac{I_F}{A_D^Sx^{(k)} + BG_F}\right)~,
\end{equation}
where an appropriate stopping rule introduces a regularization effect. In this application we utilize the same statistics-based stopping criterion introduced in \cite{bepi14} and successfully applied to the reconstruction of RHESSI images \cite{beetal13}. This approach is based on the observation that in the case of Poisson noise, standard regularization, referring to a specific data vector and concerning the point-wise convergence of a one-parameter family of regularizing operators, cannot be applied. Therefore a new definition of asymptotic regularization must be introduced, in order to account for the fact that in the Poisson case, convergence must hold when the signal-to-noise ratio of the data grows up to infinity. More formally:

\begin{definition}
\label{asym_reg_def}
An operator $R: \R^{\#F}_+ \subset \R^{\#F} \rightarrow \R^{\#S}$ is called asymptotically regularizable on a cone $\mathcal C \subset \R^{\#F}_+$ if there exists a family of continuous operators
\begin{equation}
R^{(k)} : \mathcal C \to \mathcal \R^{\#S}
\end{equation}
with $k \in \mathbb N$ (or $\mathbb R$) and 
a parameter choice rule
\begin{equation}
k: \mathbb R_+ \times \mathcal C \to \mathbb N
\end{equation}
such that, for $\delta>0$ and $I_F , I_F^\delta \in \mathcal C$
\begin{equation}
\lim_{L \to \infty} \sup \{ \| R^{(k(\delta,I_F^\delta))} (\bar I_F^\delta) - R( \bar I_F) \| ~|~ 
\|I_F^\delta - I_F\| \leq \delta ~,~ \|I_F\| > L~
\} = 0 ~ ,
\label{eq:poisson_reg_def}
\end{equation}
where $L>0$, $\bar I_F^\delta=I_F^\delta/\|I_F^\delta\|$, $\bar I_F=I_F/\|I_F\|$, and $\| \cdot \|$ is the Euclidean norm. 
Such a pair $(\{R^{(k)}\},k)$ is called an asymptotic regularization method for $R$ in $\mathcal C$.
\end{definition}

We notice that the EM algorithm can be thought as a family of operators $R^{(k)}$ by taking the concatenation of the first $k$ iterations $x^{(k)}$ and by thinking it as a function of the entry data $I_F$, i.e.
\begin{equation}
R^{(k)}(I_F,BG_F) = \underbrace{\psi_{(I_F,BG_F)} \circ \ldots \circ \psi_{(I_F,BG_F)}}_{k~\mathrm{times}} ({\bf{1}})
\end{equation}
where $\psi_{(I_F,BG_F)}$ is the EM iteration $x^{(k+1)} = \psi_{(I_F,BG_F)}(x^{(k)})$ defined in equation (\ref{EM}) and ${\bf{1}}$ the unit vector.
Since the EM algorithm is convergent we can define the limit operator with no background, i.e
\begin{equation}\label{EM-limit}
R(I_F) := \lim_{k\to\infty} R^{(k)}(I_F,0)
\end{equation}
for every $I_F \in \mathcal C$ \cite{depierro}. Now we prove that, in the presence of a given background $BG_F$, the EM algorithm is an asymptotic regularization for its limit operator $R$ when using the following stopping rule \cite{bepi14}.
\begin{definition} We call KL-KKT (Kullback Leibler - Karush Khun Tucker) stopping rule the function
\label{stoprule}
\begin{equation}\label{stoppingrule}
k(\delta,I_F^\delta) := \inf \{ k \in \mathbb N ~|~ P^{(k)}(I_F^\delta,BG_F) \leq \tau Q^{(k)}(I_F^\delta,BG_F) \} 
\end{equation}
with $\tau >0$,
\begin{equation}\label{pikappa}
P^{(k)}(I_F^\delta,BG_F) := \left\| x^{(k)} \cdot (A_D^S)^T \left( {\mathbf{1}} - \frac{I^{\delta}_{F}}{A_D^S x^{(k)}+BG_F} \right) \right\|^2_2 
\end{equation}
and
\begin{equation}
Q^{(k)}(I_F^\delta,BG_F) := \sum_{i=1}^{\#F} \left(\frac{(A_D^S)^2 (x^{(k)})^{2}}{A_D^S x^{(k)}+BG_F} \right)_i ~,
\end{equation}
where division by vectors, $(A_D^S)^2$ and $(x^{(k)})^{2}$ indicate component-wise operations and $x^{(k)}=x^{(k)}(I_F^\delta,BG_F)$ as in (\ref{EM}).
\end{definition}

We now give four technical but easy lemmas.
\begin{lem}
Let $L>0$ and $\xi \in \mathcal C$. For the EM algorithm $\psi_{(I_F,BG_F)}$ defined in equation (\ref{EM}) the following relations hold true:
\begin{equation}
\psi_{(L \cdot I_F,BG_F)}(\xi) = L ~\psi_{(I_F,BG_F)}(\xi)
\end{equation}
and
\begin{equation}
\psi_{(L \cdot I_F,BG_F)}(L \xi) = L ~\psi_{(I_F,BG_F/L)}(\xi) ~ .
\end{equation}
\end{lem}
\begin{proof}
It follows from computations.
\end{proof}
The previous lemma means that, given an initialization, say $x^{(0)}$, with positive components, for $k \geq 2$ the $k$-th EM iteration applied to signal $L \cdot I_F$ with background $BG_F$ is a multiple of the $k$-th iterate applied to $I_F$ with background $BG_F/L$, i.e. $x^{(k)}(L \cdot I_F,BG_F) = L ~ x^{(k)}(I_F,BG_F/L)$. Similar properties hold for the functions $P^{(k)}$ and $Q^{(k)}$ of Definition \ref{stoprule}:
\begin{lem}\label{almost-homogeneity}
For the KL-KKT rule defined above we have
\begin{equation}\label{almost-homogeneity1}
P^{(k)}(L \cdot I_F,BG_F) = L^2 P^{(k)}(I_F,BG_F/L)
\end{equation}
and
\begin{equation}\label{almost-homogeneity2}
Q^{(k)}(L \cdot I_F,BG_F) = L Q^{(k)}(I_F,BG_F/L)
\end{equation}
for $k \geq 2$.
\end{lem}

\begin{proof}
It follows by applying equations (\ref{almost-homogeneity1}) and (\ref{almost-homogeneity2}) to definition (\ref{stoprule}) of $P^{(k)}$ and $Q^{(k)}$.
\end{proof}

The following lemma is a readily generalization of the well-known flux preservation condition for the EM algorithm in presence of positive background.
\begin{lem} \label{Hx-ineq}
The EM method $x^{(k)}(I_F,BG_F)$ as defined in equation \ref{EM} has the following property
\begin{equation}
\sum_{i=1}^{\#F} (A_D^S x^{(k)})_i \leq \sum_{i=1}^{\#F} (I_F)_i
\end{equation}
and the equality holds only if $BG_F=0$.
\end{lem}

\begin{proof}
This follows by replacing $k$ with $k+1$ and  then by using the explicit form of the algorithm (\ref{EM}).
\end{proof}

\begin{lem}
\label{lemmaQ}
The functions $Q^{(k)}$ have positive lower and upper bounds for every pair $(I_F^\delta,BG_F)$.
\end{lem}
\begin{proof}
By using standard inequalities between norms and Lemma \ref{Hx-ineq}, we get an upper bound for $Q^{(k)}$, i.e.
\begin{eqnarray}
\label{eq:B_ineq}
Q^{(k)}(I_F^\delta,BG_F) 
& \leq & \sum_{i=1}^{\#F} \left( \frac{ (A_D^S)^2 (x^{(k)} (I_F^\delta,BG_F))^2 }{ A_D^S x^{(k)} (I_F^\delta,BG_F) }  \right)_i \nonumber \\
& \leq & \sum_{i=1}^{\#F} (A_D^S x^{(k)}(I_F^\delta,BG_F))_i \leq \sum_{i=1}^{\#F} (I_F)_i^\delta ~ ,
\end{eqnarray}
and also a lower bound
\begin{eqnarray}
\label{eq:B_ineq}
Q^{(k)}(I_F^\delta,BG_F) & \geq & \frac{\sum_{i=1}^{\#F} \left( (A_D^S)^2 (x^{(k)}(I_F^\delta,BG_F))^2 \right)_i}{\sum_{i=1}^{\#F} ((I_F^\delta)_i + \|BG_F\|_\infty)}   \nonumber \\
& \geq & \frac{\sum_{i=1}^{\#F} \left( A_D^S x^{(k)}(I_F^\delta,BG_F) \right)_i}{\#F \sum_{i=1}^{\#F} ((I_F^\delta)_i + \|BG_F\|_\infty)}   > 0 ~ ,
\end{eqnarray}
since $A_D^S x^{(k)}$ cannot tend to $0$ as the KL divergence should tend to infinity.
\end{proof}

Now we can prove the main

\begin{thm}
If the EM iteration in equation (\ref{EM}) applied to $I_F^\delta$ is stopped at the $k_*$ iterate with $k_* := k(\delta,I_F^\delta)$ according to the KL-KKT criterion (Definition~\ref{stoprule}) then $k_* < \infty$ and $x^{(k_*)}(\bar I_F^\delta,BG_F/\|I_F^\delta\|)$ tends to a solution of problem (\ref{ML-constrained}) with data $\bar I_F$ (with no background) as $\| I_F \| \to \infty$.
If, in addition, the solution is unique, then the method is an asymptotic regularization for its limit operator (\ref{EM-limit}) on the positive cone $\R^{\#F}_+$.
\end{thm}

\begin{proof}
For the Lemma \ref{lemmaQ} the function $Q^{(k)}$ in the definition (\ref{stoppingrule}) is bounded away from $0$. Moreover the function $P^{(k)}$ tends to $0$ as $k \to \infty$ by construction.
Hence, $k(\delta,I_F^\delta) < \infty$ for each $(\delta,I_F^\delta)$. Let us denote by $k_*:= k(\delta,I_F^\delta)$ the stopping index.
Since, by Lemma \ref{almost-homogeneity}, the condition in the definition (\ref{stoppingrule}) can be written as
\begin{equation} \label{k_*}
\frac{ P^{(k_*)}(\bar I_F^\delta,BG_F/\|I_F^\delta\|) }{Q^{(k_*)}(\bar I_F^\delta,BG_F/\|I_F^\delta\|)}
\leq \frac{\tau}{\|I_F^\delta\|}
\end{equation}
it results that
\begin{equation} \label{k_*}
\lim_{\|I_F\| \to \infty} P^{(k_*)}(\bar I_F^\delta,BG_F/\|I_F^\delta\|) = 0
\end{equation}
where $\|I_F^\delta\| \geq \|I_F\| - \delta$, as $Q^{(k_*)}$ is bounded. This means that $x^{(k_*)}(\bar I_F^\delta,BG_F/\|I_F^\delta\|)$ satisfies the KKT conditions as $||I_F||$ tends to infinity, which was to be shown.

If the solution of problem (\ref{ML-constrained}) with data $(\bar I_F,0)$ is unique, say $\bar x$, then $\bar x = R(\bar I_F)$ where $R(\bar I_F)$ is  the limit operator defined in (\ref{EM-limit}), and $x^{(k_*)}(\bar I_F^\delta,BG_F/\|I_F^\delta\|) \to \bar x$ as $\|I_F\| \to \infty$.

\end{proof}

We point out that EM in (\ref{EM}), together with the KL-KKT stopping rule given in Definition \ref{stoprule},  provides a regularized vector $x$ in the object space of the reconstructed images, while our aim is data de-saturation in the data space, where saturation actually occurs. Therefore our procedure ends with the projection of the KL-KKT EM reconstructed $x$ in order to construct the desaturated image $I_{desat}$. The scheme is:
\begin{equation}\label{nuova}
I_{desat} = \left\{
\begin{array}{ll}
 A_C^S x & {\mbox{in}}~ S \\
 BG_B & {\mbox{in}}~ B \\
I_F - A_D^S x &{\mbox{in}}~ F \\
I & {\mbox{in}} ~ {\mathcal{N}} \setminus (S^{\prime} \cup F)~.
\end{array}
\right.
\end{equation}
In equation (\ref{nuova}), $A_C^S: \R^{\#S} \rightarrow \R^{\#S}$ is the sub-matrix of $A_C$ that maps the vector of the values of $x$ in $S$ onto the corresponding vector in $S$ and $BG_B$ is the restriction to the blooming region $B$ of the background estimated as described in the Appendix.

\section{Numerical validation}
In order to validate the effectiveness of this approach for restoring information in the primary saturation region we consider two simulations that mimic the presence of primary saturation in {\em{SDO/AIA}} data with two different levels of adherence with experimental conditions (the way blooming can be eliminated from experimental images is illustrated in the Appendix and some examples are given in the next Section). 

\begin{table}
\begin{center}
\begin{tabular}{|c|c|c|}
\hline
peak value $[10^4]$  & $\sigma$ $[{\mbox{arcsec}}]$ & position $[{\mbox{arcsec}}]$ \\
\hline\hline
$4$ & $1.5$ & $(-20,-5)$ \\
\hline
$3$ & $2.5$ & $(5,5)$ \\
\hline
$5$ & $1.0$ & $(-2,7)$ \\
\hline
\end{tabular}
\caption{The physical and geometric parameters associated to the synthetic sources used in the first validation test for the inverse diffraction approach.}
\label{table:tab-1}
\end{center}
\end{table}

In the first simulation, the configuration that generates the synthetic data is represented by three two-dimensional symmetric Gaussian functions whose geometrical characteristics are described in Table \ref{table:tab-1}. We added a constant offset to this simulated configuration in order to mimic the presence of background and numerically convolved the resulting image with the global {\em{AIA}} PSF $A$ in equation (\ref{psf}). The resulting blurred and diffracted image (see Figure \ref{fig:simulation-2}, top left panel) was affected by Poisson noise and the result was artificially saturated by setting to $16383$ DN pixel$^{-1}$ all grey levels greater than this value (see Figure \ref{fig:simulation-2}, top right panel). We finally applied our numerical automatic procedure described in Section 3 and obtained the result in Figure \ref{fig:simulation-2}, bottom left panel. Finally, the bottom right panel shows the level of accuracy with which the reconstruction method is able to recover information inside the saturated region of the image (the root mean square error in the saturation region is of the order of $9\%$).

\begin{figure}
\begin{tabular}{cc}
\includegraphics[width=0.5\textwidth]{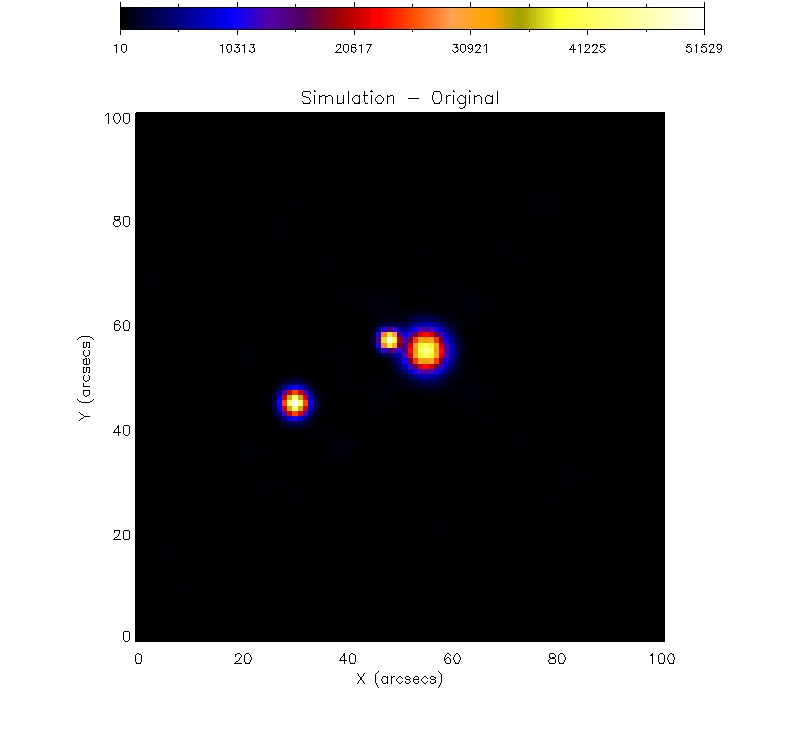}&
\includegraphics[width=0.5\textwidth]{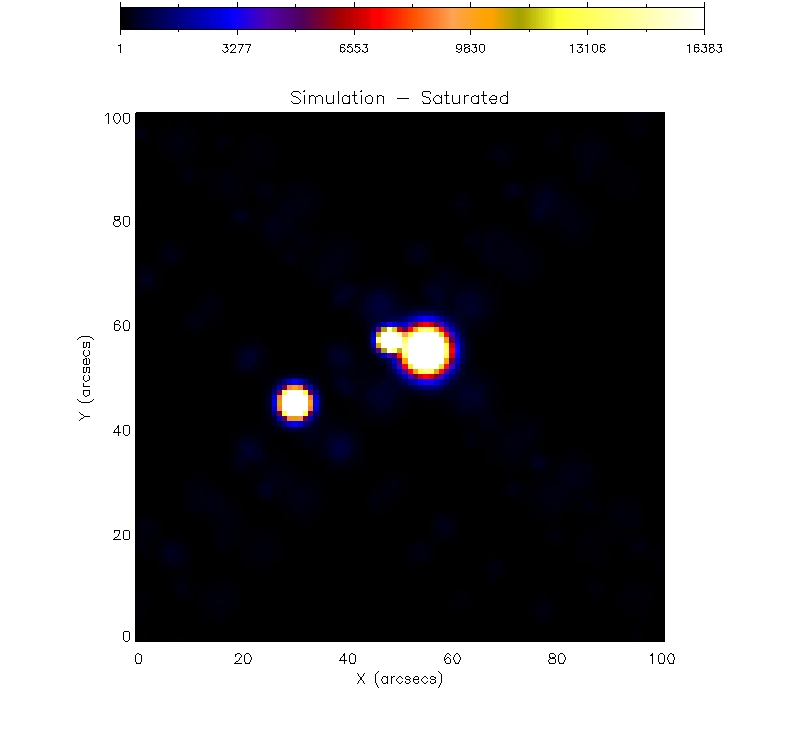}\\
\includegraphics[width=0.5\textwidth]{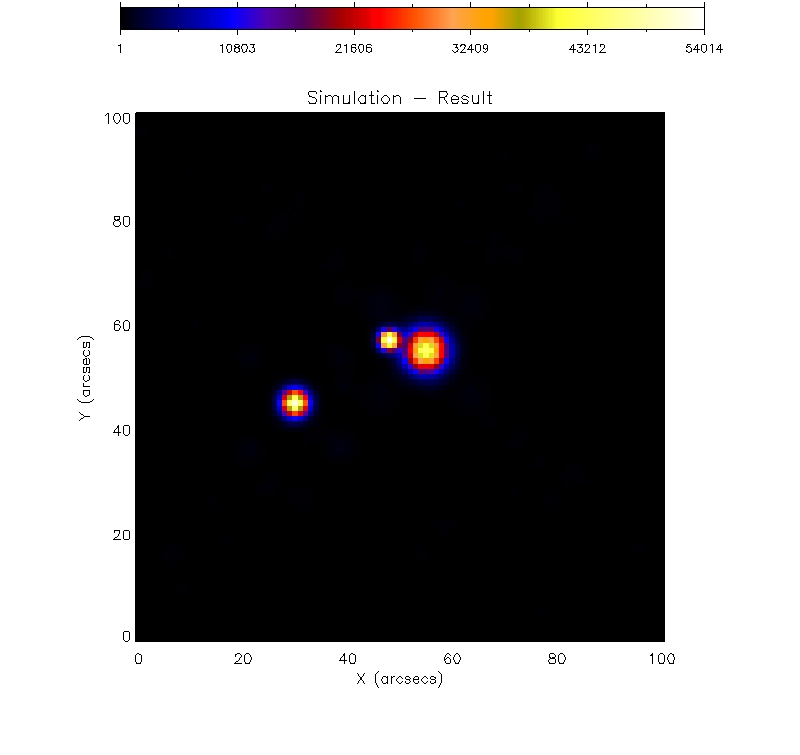}&
\includegraphics[width=0.45\textwidth]{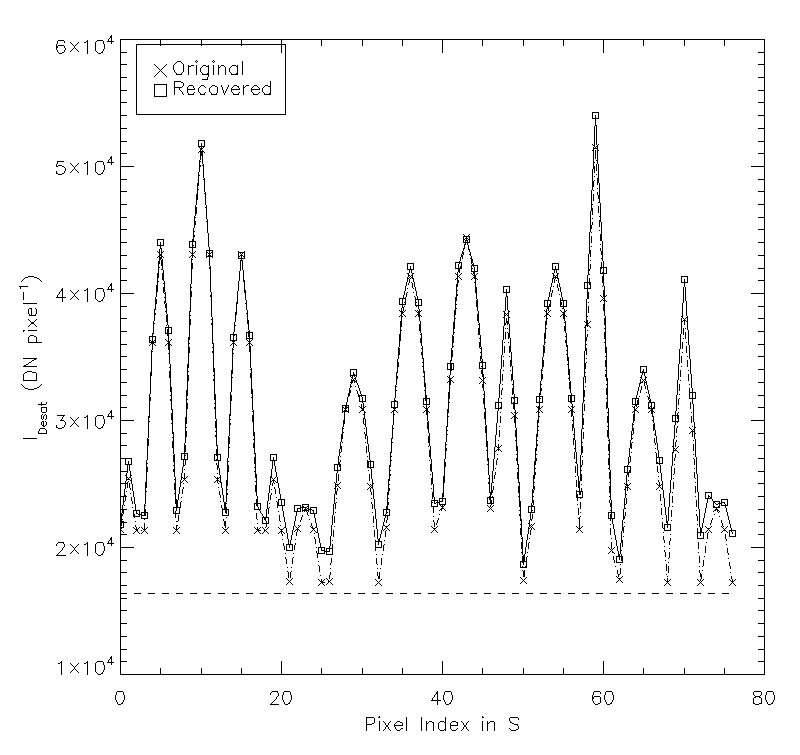}\\
\end{tabular}
\caption{Validation of the correlation/inversion method against synthetic data. Top left panel: the ground-truth image obtained from the synthetic configuration by applying diffraction and blurring. Top right panel: artificially saturated image with Poisson noise added. Bottom left panel: de-saturated image. Bottom right panel: lexicographic representation of the grey-level values for the pixels in the saturated region: ground truth values (dashed line); saturated values (solid line); reconstructed values (black line).}
\label{fig:simulation-2}
\end{figure}

The second simulation created a synthetic saturated dataset starting from the non-saturated real {\em{AIA}} $4096 \times 4096$ image in Figure \ref{fig:simulated-3}, top left panel. The simulation process is implemented by means of the following steps:
\begin{enumerate}
\item The diffraction fringes and the diffusion blurring are eliminated by applying a deconvolution step based on EM (with the global PSF), which provides the map in Figure \ref{fig:simulated-3}, top right panel. 
\item The pixels in the brighter part of the image have been carried over the saturation threshold, to obtain the image in Figure \ref{fig:simulated-3}, middle left panel. To do this we applied the following re-scaling procedure: denoting with ${\tilde{x}}$ the pixel intensity before the rescaling, the rescaled pixel intensity is 
\begin{equation}\label{re-scaling}
{\tilde{x}}_{rescaled}=\left\{
\begin{array}{lc}
\frac{mM - {\tilde{x}}^*}{M - {\tilde{x}}^*} {\tilde{x}} + \frac{M(1-m)}{M-{\tilde{x}}^*} {\tilde{x}}^* & {\tilde{x}} \ge {\tilde{x}}^*\\
{\tilde{x}}  & {\tilde{x}} < {\tilde{x}}^*\\
\end{array}\right.
\end{equation}
where ${\tilde{x}}^*= 0.25 M$, $M$ is the maximum intensity in the image and $m=12$ in the figure.
\item The image in Figure \ref{fig:simulated-3}, middle left panel, is first convolved with the core PSF $A_C$ to construct the prototype, in Figure \ref{fig:simulated-3}, middle right panel, of the ideal image that would be recorded by {\em{AIA}} if no diffraction and no saturation occurred. This is the ground-truth we want to restore.
\item The same image in Figure \ref{fig:simulated-3}, middle left panel, is now convolved with the global PSF $A$ and the result is first affected by Poisson noise and then saturated by setting to $16383$ DN pixel$^{-1}$ all the pixel values over the saturation threshold. The result of this step is in Figure \ref{fig:simulated-3}, bottom left panel.
\item Finally, we apply the de-saturation method described in Section 3 to this image, to obtain the reconstruction in Figure \ref{fig:simulated-3}, bottom right panel.
\end{enumerate}
In order to quantitatively assess the reliability of this procedure we have computed, for different values of $m$, the C-statistics
\begin{equation}\label{c-stat}
C_{stat}(I_F,x) = \frac{2}{\#F} \sum_{i=1}^{\# F} (I_F)_i \log\frac{(I_F)_i}{(A_D^S x + BG_F)_i} + (A_D^S x+BG_F)_i - (I_F)_i~,
\end{equation}
that measures, according to the Kullbach-Leibler topology, the discrepancy between the data in the region of the diffraction fringes and the expectation corresponding to the reconstructed source. Table \ref{table:cstat} shows such $C_{stat}$ values; computes the root mean square errors in $S$ between the ground truth and the de-saturated images; and, finally, compares the sum $T_F$ of the values of the original flux $I_F$ from the pixels in the diffraction fringes, with the sum $T^{pred}_{F}$ of the values of the flux in the diffraction fringes predicted by the reconstructed source, again for different values of $m$ in equation (\ref{re-scaling}).



\begin{figure}[pht]
\begin{center}
\begin{tabular}{cc}
\includegraphics[width=6.cm]{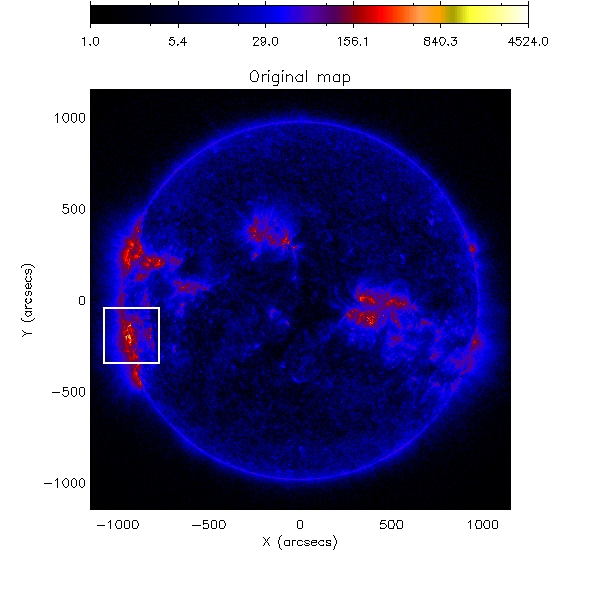} & 
\includegraphics[width=6.cm]{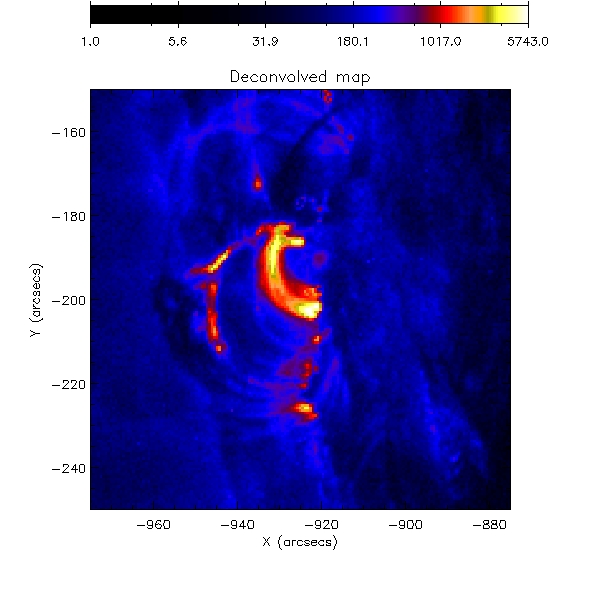} \\
\includegraphics[width=6.cm]{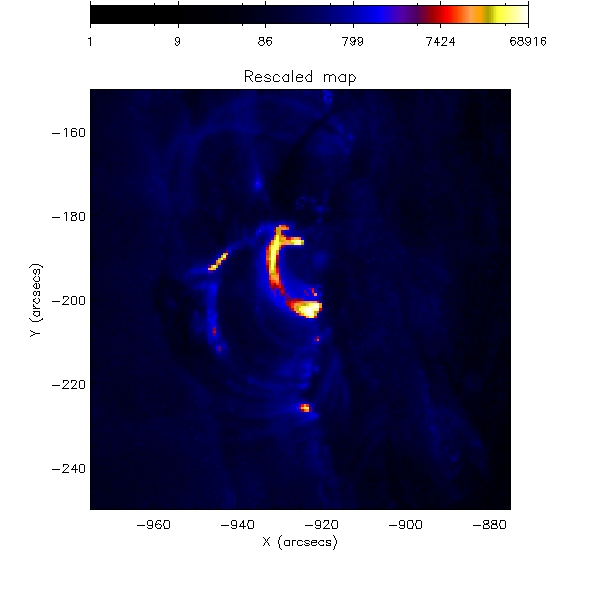} &
\includegraphics[width=6.cm]{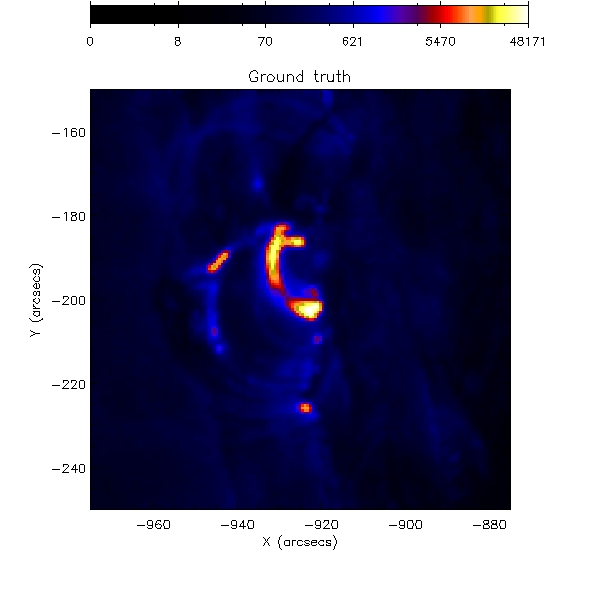} \\ 
\includegraphics[width=6.cm]{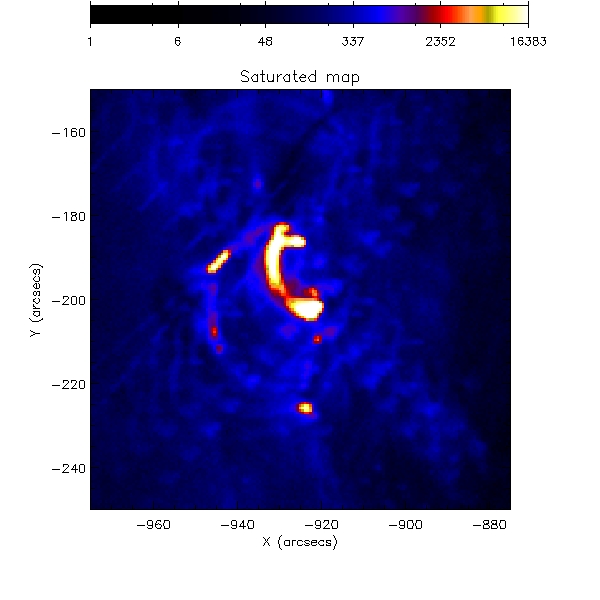} & 
\includegraphics[width=6.cm]{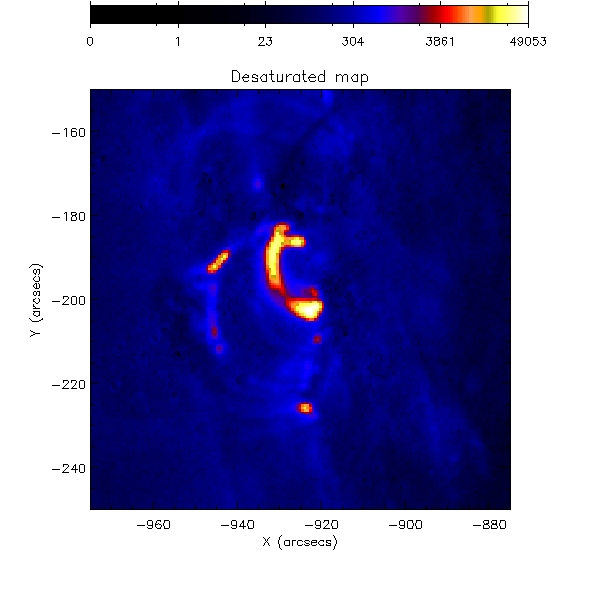} 
\end{tabular}
\caption{The de-saturation method at work in the case of an experimental image synthetically saturated. Top left panel: the original image with highlighted the region to de-convolve. Top right panel: the deconvolved image. Middle left panel: re-scaled image. Middle right panel: prototype of the ideal un-diffracted, un-saturated image. Bottom left panel: saturated image affected by Poisson noise. Bottom right panel: de-saturated image. The images to compare are this last one and the one in the middle right panel.}
\label{fig:simulated-3}
\end{center}
\end{figure}

\begin{table}
\begin{center}
\begin{tabular}{|c|c|c|c|c|}
\hline\hline
$m$ & $C_{stat}$ & $T_F (10^6) $ & $T^{pred}_{F} (10^6)$ & RMS (\%) \\
\hline\hline
$6$ & $4.66$ & $3.18$ & $3.26$ & $4.36$ \\
\hline
$9$ & $3.11$ & $4.33$ & $4.43$ & $1.72$ \\
\hline
$12$ & $2.29$ & $5.99$ & $6.06$ & $0.95$ \\
\hline
$15$ & $2.08$ & $6.78$ & $6.82$ & $0.56$ \\
\hline
$18$ & $1.67$ & $7.51$ & $7.52$ & $0.36$ \\
\hline\hline
\end{tabular}
\caption{The de-saturation method at work in the case of an experimental image synthetically saturated: C-statistic, original and reconstructed flux intensity in correspondence with the diffraction fringes, and root mean square error for several re-scaling intensities.}
\label{table:cstat}
\end{center}
\end{table}

\section{Application to real data}
The processing and, specifically, the de-saturation of {\em{SDO/AIA}} data is truly a big data issue: indeed, {\em{AIA}} data include images of the Sun in $7$ EUV wavelengths every $12$ seconds since February 2010. Here we just focused on a small set of examples and, in particular, we considered the single event occurred on September 6 2011, detected at the four different wavelengths $94 \AA$, $131 \AA$, $171 \AA$, and $193 \AA$, around the same acquisition time, i.e. 22:18:50 UT, 22:19:25 UT, 22:16:48 UT, and 22:16:43 UT, respectively. Figure \ref{fig:real-1} compares the original saturated images with the reconstructed ones. We used the correlation/inversion process described in Section 3 in order to recover the information lost due to primary saturation and the interpolation procedure described in the Appendix to reduce the effect of blooming. The figure visually demonstrates the effectiveness of the de-saturation process. However, in order to provide a quantitative assessment of such effectiveness, in Table {\ref{table:c-stat-real}} we provide the C-statistic values in correspondence of the diffraction fringes for the four different wavelengths.

\begin{table}
\begin{center}
\begin{tabular}{|c|c|c|c|c|}
\hline\hline
wavelength ($\AA$) & time (UT) & $T_F (10^6)$ & $T_{F}^{pred} (10^6)$ & $C_{stat}$ \\
\hline\hline
$94$ & 22:18:50 & $3.44$ & $3.50$ & $2.83$ \\
\hline
$131$ & 22:19:25 & $4.51$ & $4.55$ & $1.97$ \\
\hline
$171$ & 22:16:48 & $6.73$ & $6.76$ & $3.66$ \\
\hline
$193$ & 22:16:43 & $1.16$ & $1.15$ & $8.97$ \\
\hline\hline
\end{tabular}
\caption{The de-saturation method at work in the case of an experimental image recording during the September 6 2011 event. The fluxes are computed in correspondence with the diffraction fringes.}
\label{table:c-stat-real}
\end{center}
\end{table}

Finally, Figure \ref{fig:final} shows the effectiveness of the method in the case of data recorded at a time point when the saturation effect was really dramatic. However the de-saturation method shows once more its power in both the reconstruction of the primary information and in the data interpolation for the blooming region. It is interesting to note that, for this case, the C-statistic value is higher ($\sim 14.03$) than for the cases described in Table {\ref{table:c-stat-real}}. This is most likely due to the fact that our algorithm does not account for the wavelength-dependent dispersion of the {\em{AIA}} PSFs (the formulation of a multi-wavelength approach to this de-saturation problem is in progress).

\begin{figure}[pht]
\begin{center}
\begin{tabular}{cc}
\includegraphics[width=5.cm]{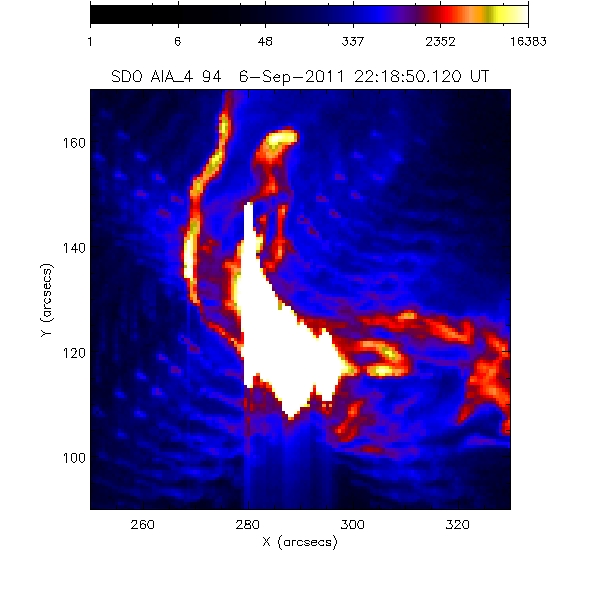} &
\includegraphics[width=5.cm]{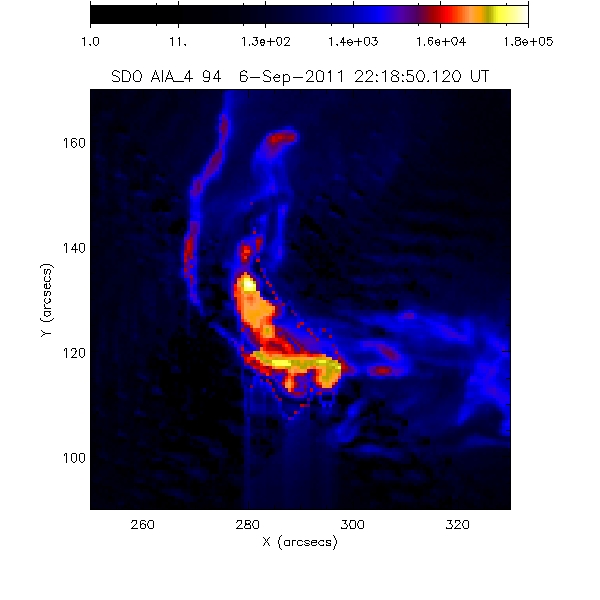} \\
\includegraphics[width=5.cm]{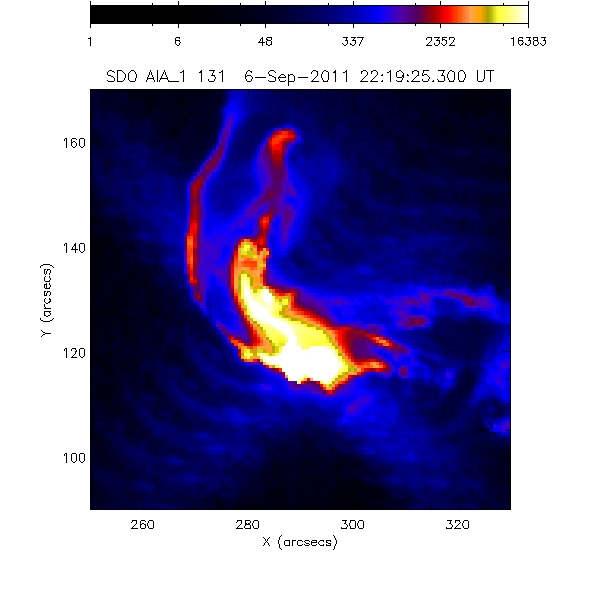} &
\includegraphics[width=5.cm]{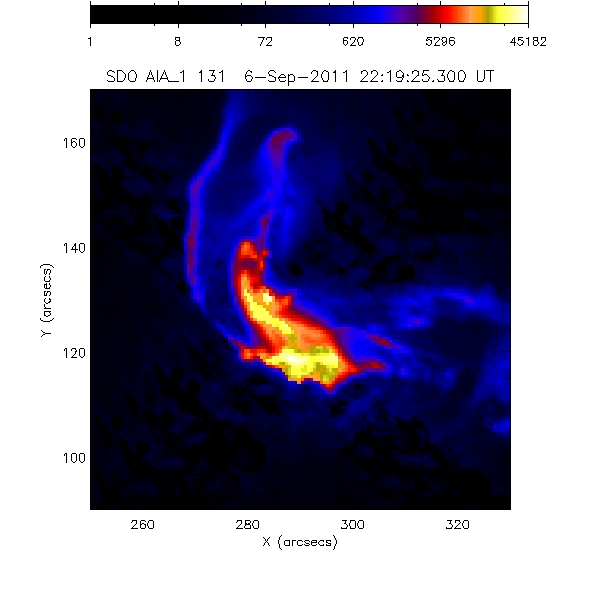} \\
\includegraphics[width=5.cm]{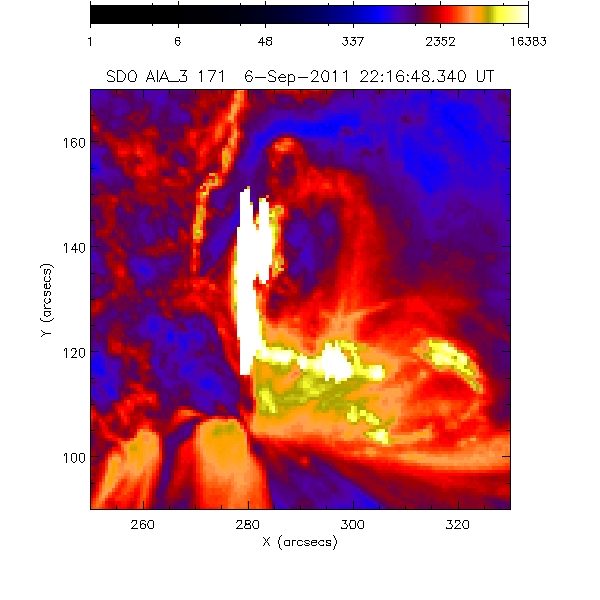} &
\includegraphics[width=5.cm]{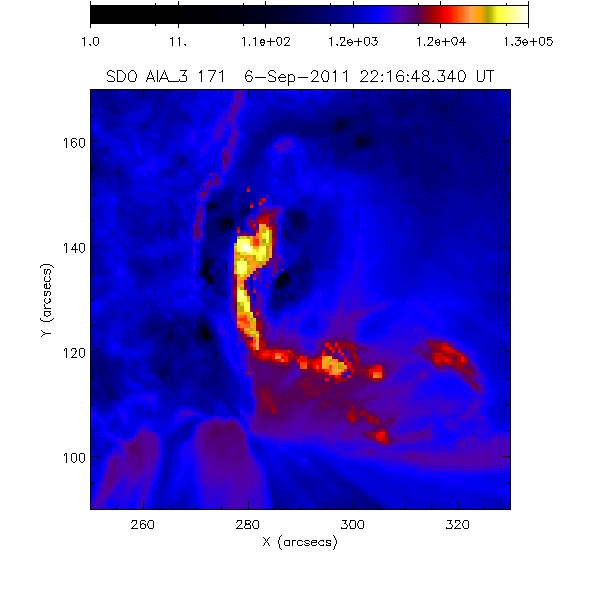} \\
\includegraphics[width=5.cm]{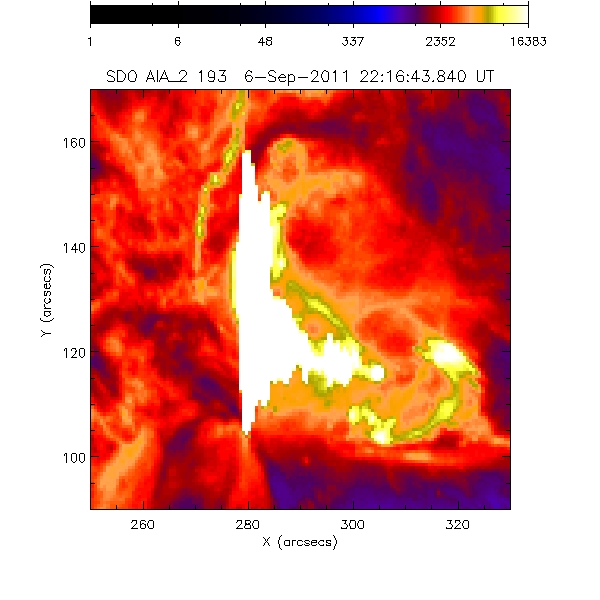} &
\includegraphics[width=5.cm]{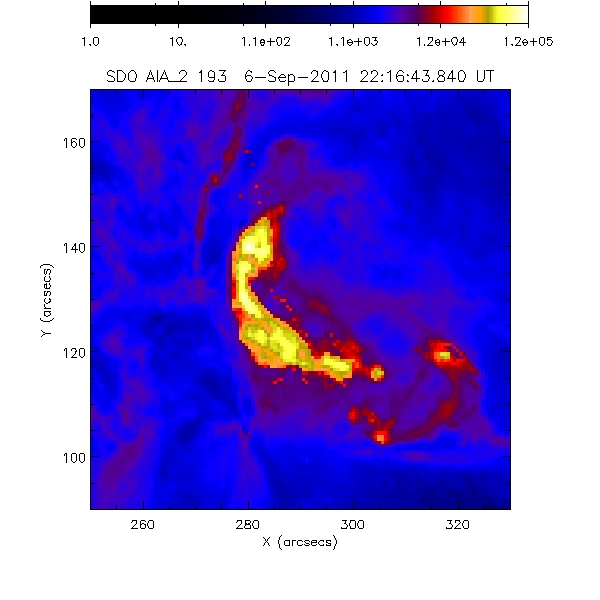} \\
\end{tabular}
\caption{The de-saturation method at work for experimental data recorded during the September 6, 2011 event. First row panels: experimental and de-saturated images for the $94 \AA$ bandwidth at 22:18:50~UT. Second row panels: experimental and de-saturated images for the $131 \AA$ bandwidth at 22:19:25~UT. Third row panels: experimental and de-saturated images for the $171 \AA$ bandwidth at 22:16:48~UT. Fourth row panels: experimental and de-saturated images for the $193 \AA$ bandwidth at 22:16:43~UT.}
\label{fig:real-1}
\end{center}
\end{figure}

\begin{figure}[pht]
\begin{center}
\begin{tabular}{cc}
\includegraphics[width=7.cm]{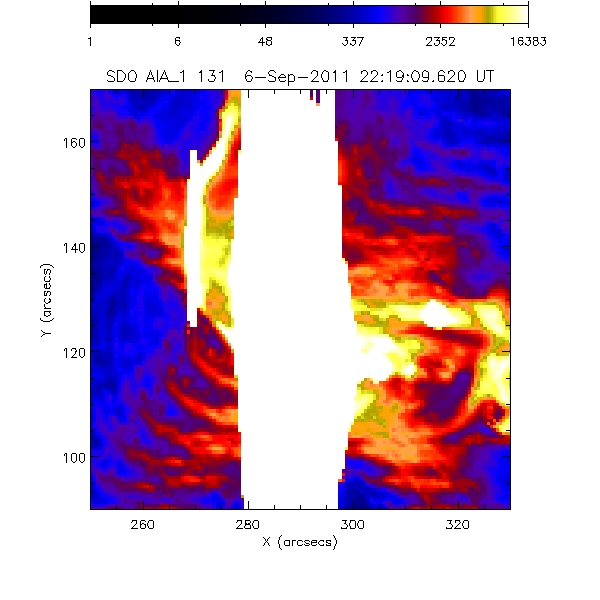} &
\includegraphics[width=7.cm]{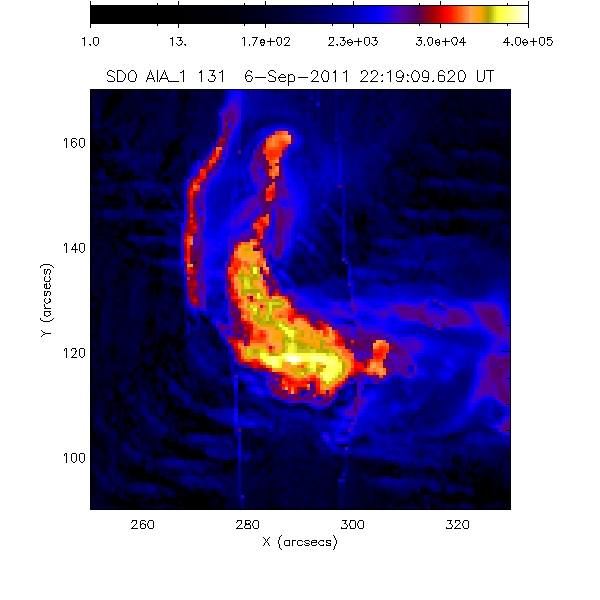} \\
\end{tabular}
\caption{The case of the data collected at 22:19:09 UT, on September 6 2011, by means of the $131 \AA$ passband. In this case the saturation effects are really impressive but the method is still able to recover the information in the saturated region.}
\label{fig:final}
\end{center}
\end{figure}

%
%

\section{Conclusions}
{\em{SDO/AIA}} images are strongly affected by both primary saturation and blooming, that may occur at all different wavelengths and acquisition times, even in the case of flaring events characterized by a moderate peak flux. This paper describes the first mathematical description of a robust method for the de-saturation of such images at both a primary and secondary (blooming) level. The method relies on the description of de-saturation in terms of inverse diffraction and utilizes correlation and Expectation-Maximization for the recovery of information in the primarily saturated region. This approach requires to compute a reliable estimate of the image background which, for this paper, has been obtained by means of interpolation in the Fourier space. The knowledge of the background permits to recover information in the blooming region in a very natural way. 

The availability of an automatic procedure for image de-saturation in the {\em{SDO/AIA}} framework may potentially change the extent with which EUV information from the Sun can be exploited. In fact, armed with our computational approach, many novel problems can be addressed in {\em{SDO/AIA}} imaging. For example, one can study the impact of the choice of the model for the diffraction PSF on the quality of the de-saturation. In this paper we used a synthetic estimate of the diffraction PSF provided by Solar SoftWare (SSW) but other empirical or semi-empirical forms can be adopted. Furthermore, this technique can be extended to account for the dependance of the PSF from the passband wavelengths. Finally, the routine implementing this approach is fully automated and this allows the systematic analysis of many events recorded by {\em{AIA}} and their integration with data provided by other missions such as {\em{RHESSI}} \cite{lietal02} or, in the near future {\em{STIX}} \cite{beetal12}.

\section*{Acknowledgements}
This work was supported by a grant of the Italian INdAM - GNCS and by the NASA grant NNX14AG06G.

\section*{Appendix}
The {\em{SDO/AIA}} hardware is equipped with a feedback system that reacts to saturation in correspondence of intense emissions by reducing the exposure time. As a result, for a typical {\em{AIA}} acquisition along a time range of some minutes, during which saturation occurs, the telescope always provides some unsaturated frames that can be utilized to estimate the background. A possible way to realize such an estimate is based on the following scheme. Let us denote with $I_1$ and $I_3$ two unsaturated images acquired at times $t_1$ and $t_3$, respectively and with $I_2$ a saturated image acquired at $t_2$ with $t_1 < t_2 < t_3$ (note that $I_1$, $I_2$, and $I_3$ are normalized at the same exposure time). The algorithm for the estimate of the background is:
\begin{enumerate}
\item $I_1$ and $I_3$ are deconvolved with EM (using the global PSF) to obtain the reconstructions ${\tilde{x}}_1$ and ${\tilde{x}}_3$ (the KL-KKT rule can be used to stop the iterations).
\item Both ${\tilde{x}}_1$ and ${\tilde{x}}_3$ are Fourier transformed by means of a standard FFT-based procedure to obtain $\hat{{\tilde{x}}}_1$ and $\hat{{\tilde{x}}}_3$.
\item A low-pass filter is applied to both $\hat{{\tilde{x}}}_1$ and $\hat{{\tilde{x}}}_3$ to obtain $\hat{{\tilde{x}}}^f_1$ and $\hat{{\tilde{x}}}^f_3$, respectively.
\item For each corresponding pair of pixels in $\hat{{\tilde{x}}}^f_1$ and $\hat{{\tilde{x}}}^f_3$ that are not negligible, an interpolation routine is applied, both for the real and imaginary part. This provides $\hat{{\tilde{x}}}^{int}_{2}$ in correspondence of $t=t_2$.
\item The resulting vector $\hat{{\tilde{x}}}^{int}$ is Fourier inverted to obtain the interpolated reconstruction ${\tilde{x}}^{int}_2$.
\item The core PSF $A_C$ is finally applied to ${\tilde{x}}^{int}_2$ to obtain $I^{int}_2$ in the image domain.
\end{enumerate} 
$I^{int}_2$ is a reliable estimate, for time $t_2$, of the background $BG=A_C {\tilde{x}}$ introduced in Section 2. On the other hand, a reliable estimate of the image in the bloomed region is provided by the restriction $BG_B$ of $I^{int}_2$ onto $B$ determined as in Section 2. We finally observe that, in this algorithm, the interpolation step is applied in the Fourier domain because, after filtering, a lot of pixels are negligible and therefore the computational burden of the procedure is notably decreased.



\end{document}